\documentclass[11pt]{article}
\usepackage{amsmath,amsthm}
\usepackage{amssymb,latexsym}
\usepackage{epsfig}
\usepackage{hyperref}
\usepackage{float}
\usepackage{mathpazo}
\usepackage{fullpage}

\newtheorem{theorem}{Theorem}[section]

\newtheorem{claim}[theorem]{Claim}
\newtheorem{lemma}[theorem]{Lemma}

\newtheorem{definition}[theorem]{Definition}
\newtheorem{remark}[theorem]{Remark}

\floatstyle{ruled}
\newfloat{program}{!h}{loa}
\floatname{program}{SDP}
\newfloat{problem}{h}{loa}
\floatname{problem}{Problem}
\newfloat{algorithm}{h}{loa} 
\floatname{algorithm}{Algorithm}
\newfloat{measurement}{h}{loa}
\floatname{measurement}{Measurement}

\newcommand{\qedsymb}{\hfill{\rule{2mm}{2mm}}}
\renewenvironment{proof}[1][]{\begin{trivlist} 
\item[\hspace{\labelsep}{\bf\noindent Proof#1:\/}] }{\qedsymb\end{trivlist}}

\def\Z{{\mathbb{Z}}}

\def\omegans{{\omega^{ns}}}
\DeclareMathOperator{\Exp}{{\rm Exp}}

\newcommand\ip[1]{{\left\langle {#1} \right\rangle}}

\newcommand{\ignore}[1]{}

\newcommand{\eps}{\varepsilon}
\renewcommand{\epsilon}{\varepsilon}

\begin{document}

\title{\bf No Strong Parallel Repetition with Entangled and Non-signaling Provers}

\author{
Julia Kempe\footnote{Blavatnik School of Computer Science, Tel Aviv University, Tel Aviv 69978, Israel.
Supported by the European Commission under the Integrated Project Qubit Applications (QAP) funded
by the IST directorate as Contract Number 015848, by an Alon Fellowship of the Israeli Higher
Council of Academic Research, by an Individual Research Grant of the Israeli Science Foundation, by
a European Research Council (ERC) Starting Grant and by a Raymond and Beverly Sackler Career
Development Chair.}
 \and
 Oded Regev\footnote{Blavatnik School of Computer Science, Tel Aviv University, Tel Aviv 69978, Israel. Supported
   by the Binational Science Foundation, by the Israel Science Foundation,
   by the European Commission under the Integrated Project QAP funded by the IST directorate as Contract Number 015848, and by a European Research Council (ERC) Starting Grant.}
}

\date{}

\maketitle

\begin{abstract}
We consider one-round games between a classical verifier and two provers. One of the main questions in this area is the
\emph{parallel repetition question}: If the game is played $\ell$ times in parallel, does the maximum winning probability
decay exponentially in $\ell$? In the classical setting, this question was answered in the affirmative by Raz.
More recently the question arose whether the decay is of the form $(1-\Theta(\eps))^\ell$
where $1-\eps$ is the value of the game and $\ell$ is the number of repetitions.
This question is known as the \emph{strong parallel repetition question} and was motivated
by its connections to the unique games conjecture.
It was resolved by Raz who showed that strong parallel repetition does \emph{not}
hold, even in the very special case of games known as XOR games.

This opens the question whether strong parallel repetition holds in the case when the provers share entanglement.
Evidence for this is provided by the behavior of XOR games, which have strong (in fact \emph{perfect}) parallel repetition,
and by the recently proved strong parallel repetition of linear unique games. A similar question was open for
games with so-called non-signaling provers. Here the best known
parallel repetition theorem is due to Holenstein, and is of the form $(1-\Theta(\eps^2))^\ell$.

We show that strong parallel repetition holds neither with entangled provers nor with non-signaling provers.
In particular we obtain that Holenstein's bound is tight. Along the way we also provide a tight characterization
of the asymptotic behavior of the entangled value under parallel repetition of unique games
in terms of a semidefinite program.

\end{abstract}

\newpage

\section{Introduction}

\paragraph{Games:}
Two-prover games play a major role both in theoretical computer science, where they led to
many breakthroughs such as the discovery of tight inapproximability results \cite{Hastad01},
and in quantum physics, where already for more than half a century they are
used as a way to understand and experimentally verify quantum mechanics.
In such games, a verifier (or referee) chooses two questions, and sends
one question to each of two non-communicating and computationally unbounded provers (or players) who then respond with answers
taken from $\{1,\ldots ,k\}$ for some $k \geq 1$. The verifier decides
whether to accept (or in other words, whether the players win the game). The question we ask
is: given the verifier's behavior as specified by the game, what is the maximum winning probability
of the provers?

It turns out that the answer to this question depends on the exact power we give to the provers.
In the model most commonly used in theoretical computer science (the classical model),
the provers are simply deterministic functions of their inputs. We call the
maximum winning probability in this case the \emph{(classical) value} of the game
and denote it by $\omega$. We could also allow the provers to share randomness,
but it is easy to see that this cannot increase their winning probability.
Consider, for instance, the CHSH game~\cite{Clauser:69a}: Here the verifier chooses two random bits $x$ and $y$, and sends one to each
prover; he then receives as an answer one bit from each prover (so $k=2$ here), call them $a$ and $b$.
The verifier accepts iff $a \oplus b=x \wedge y$. It is not hard to see that the value
of this game is $\omega(CHSH) = 3/4$.

The second model we consider is that of \emph{entangled provers} in which the two provers,
who still cannot communicate, are allowed to use shared entanglement.
These games, which are sometimes called {\em nonlocal games} in the physics literature, have their origins in the seminal
papers by Einstein, Podolsky, and Rosen~\cite{EinsteinPR35} and Bell~\cite{Bell:64a}.
We define the {\em entangled value} of a game
as the maximum success probability achievable by provers that share entanglement, and denote it by $\omega^*$.
Notice that by definition, for any game we have $\omega^* \ge \omega$.
One of the most astonishing features of quantum mechanics is that sharing entanglement gives the
provers the remarkable ability to create correlations that are impossible to obtain
classically, and hence increase their winning probability.
For instance, it can be shown that $\omega^*(CHSH)=(2+\sqrt{2})/4\approx 0.85$.
This gap between the classical value and the entangled value has fascinated physicists for
decades, and is used as an experimental way to validate quantum mechanics.
Another example of such a gap appears in the so-called \emph{odd-cycle game},
in which, roughly speaking, the provers are asked to color the vertices
of a cycle of length $n$ for some odd $n \ge 1$ with two colors in
such a way that the two colors adjacent to each edge are different.
The value of this game is $1-1/2n$, whereas its entangled value is $1-\Theta(1/n^2)$.

The third model we consider in this paper is that of \emph{non-signaling provers}. This model is of interest mainly as
a theoretical tool to understand the other two models (see, e.g., \cite{ItoKM09,Toner:LP}), as well as two-prover games
in general physical theories~\cite{Linden:Nonlocal} (see also \cite{Brassard:Nonlocal}).
Here, the provers can choose for any question pair an arbitrary distribution on the answers, with the only constraint
being the \emph{non-signaling constraint} --- namely, that the marginal distribution of each prover's answer must only
depend on the question to that prover (and not on the other prover's question). This constraint captures the physical
requirement that the provers are unable to communicate, and leads to the definition of the {\em non-signaling value} of
a game, which we denote by $\omegans$. Notice that for any game, $\omegans \ge \omega^* \ge \omega$. For instance, it is not
hard to see that $\omegans(CHSH)=1$, since we can arrange the distributions on the answers in such a way that the
marginal distributions are always uniform, and at the same time only winning answers are returned.

An important special case of two-prover games is that of {\em unique games}.
Here, the verifier's decision is restricted to be of the form $b=\sigma(a)$
for some permutation $\sigma$ on $[k]$. If, moreover, $k=2$, then the game is
called an \emph{XOR game}. An example of such a game is the CHSH game.
It is very common for the answer set $[k]$ in a unique game to be identified
with some group structure (e.g., $\Z_k$) and for the verifier to check
whether the difference of the two answers $a-b$ is equal to some value.
If this is the case, then we refer to the game as a \emph{linear game}.
In recent years, unique games became one of the most heavily studied topics in theoretical computer
science due to Khot's unique games conjecture~\cite{Khot02} and its strong implications
for hardness of approximation (see, e.g.,~\cite{KhotKMO07}).

\paragraph{Parallel repetition:}
One of the main questions in the area of two-prover games is the \emph{parallel repetition question}.
Here we consider the game $G^\ell$ obtained by playing the game $G$ in parallel $\ell$ times.
More precisely, in $G^\ell$ the verifier sends $\ell$ independently chosen
question pairs to the provers, and expects as answers elements of $[k]^\ell$. He accepts iff all $\ell$ answers
are accepted in the original game. It is easy to see that $\omega(G^\ell) \ge \omega(G)^\ell$
since the provers can play their optimal strategy for $G$ on each of the $\ell$ question
pairs. Similarly, $\omega^*(G^\ell) \ge \omega^*(G)^\ell$ and $\omegans(G^\ell) \ge \omegans(G)^\ell$.
Although at first it might seem that equality should hold here, the surprising fact
is that in most cases the inequality is strict. Even for a simple game like
CHSH we have that $\omega(CHSH^2)=5/8$ (which is bigger than the $9/16$ one
might expect).

The parallel repetition question asks for upper bounds on the value of repeated games.
This fundamental question has many important implications, most notably to
tight hardness of approximability results (e.g.,~\cite{Hastad01}).
The first dramatic progress in this area was made by Raz~\cite{Raz98},
with more recent work by Holenstein~\cite{Holenstein07} and Rao~\cite{Rao:parallel}.
The following theorem summarizes the state of the art in this area.
\begin{theorem}\label{th:parallel}
Let $G$ be a two-prover game with answer size $k$ and value $\omega(G)=1-\eps$. Then for all $\ell \ge 1$,
\begin{enumerate}
\item \cite{Holenstein07} $\omega(G^\ell) \leq (1-\eps^3)^{\Omega{(\ell/\log k)}}$;
 \item \cite{Rao:parallel} If $G$ is a projection game (which is a more general class than unique games)
 then $\omega(G^\ell) \leq (1-\eps^2)^{\Omega{(\ell)}}$.
\end{enumerate}
\end{theorem}

In an attempt to better understand the unique games conjecture, Feige, Kindler, and O'Donnell \cite{FKR:foams} asked
whether the bound on $\omega(G^\ell)$ above can be improved to $(1-\eps)^{\Omega{(\ell)}}$, a result called
\emph{strong parallel repetition}. Given the improved bound by Rao, it is only natural to hope that the exponent could
be lowered all the way down to $1$. They observed that if such a result holds, even just for unique games, then we would
get an equivalence of the unique games conjecture to other better studied problems like MAX-CUT.

Somewhat surprisingly, Raz~\cite{Raz08} showed that strong parallel repetition does \emph{not} hold in general. He
showed an example of an XOR game (which is no other than the odd-cycle game mentioned above) whose value is $1-1/2n$
yet even after $n^2$ repetitions, its value is still at least some positive constant. Raz's example was further
clarified and generalized in \cite{BarakHHRRS08} by showing a connection between $\omega(G^\ell)$ and the value of a
certain SDP relaxation of the game. We mention that strong parallel repetition is known to hold in the case of
projection games that are \emph{free}, i.e., the distribution on the questions to the provers is a product
distribution~\cite{BarakRRRS09}. See also~\cite{AroraKKSTV08} for an ``almost strong" parallel repetition statement for
unique games played on expander graphs.


Parallel repetition is much less well understood in the case of entangled provers. In fact, no parallel repetition result
is known for the entangled value of general games, and this is currently one of the main open questions in the area.
However, parallel repetition results \emph{are} known for several classes of games with entangled provers, as described
in the following theorem, in which we also mention Holenstein's~\cite{Holenstein07} parallel repetition
result for the non-signaling value.

\begin{theorem}\label{th:parallelforentangled}
Let $G$ be a two-prover game with answer size $k$, entangled value $\omega^*(G)=1-\eps^*$, and non-signaling value
$\omega^{ns}(G)=1-\eps^{ns}$. Then for all $\ell \ge 1$,
\begin{enumerate}
 \item \cite{CleveSUU07} If $G$ is an XOR game, then $\omega^*(G^\ell) =(1-\eps^*)^\ell$;
 \item \cite{KempeRT:unique} If $G$ is a unique game, then $\omega^*(G^\ell) \leq (1-\frac{(\eps^*)^2}{16})^\ell$;
 \item \cite{KempeRT:unique} If $G$ is a linear game, then $\omega^*(G^\ell) \leq (1-\frac{\eps^*}{4})^\ell$;
 \item \cite{Holenstein07}  $\omegans(G^\ell) \leq (1-(\eps^{ns})^2)^{\Omega{(\ell)}}$.
\end{enumerate}
\end{theorem}

Hence, we see that strong parallel repetition holds for the entangled value of linear games. In fact, in the case of
XOR games, we have \emph{perfect} parallel repetition.

All the above results involving the entangled value are derived by (i)
showing that $\omega^*$ is close (or in fact \emph{equal} in the case of XOR games)
to a certain SDP relaxation (which appears as SDP1 below), and (ii) showing
that this SDP relaxation ``tensorizes", i.e., that the value of the SDP corresponding
to $G^\ell$ is exactly the $\ell$th power of the value of the
SDP corresponding to $G$.

The above naturally raises the question of whether the entangled value obeys \emph{strong} parallel repetition, if not
in the general case, then at least in the case of unique games. The nearly tight characterization of the entangled
value of unique games using semidefinite programs \cite{KempeRT:unique} (see Lemma~\ref{lem:SDP2rounding} below) is one
reason to hope that such a strong parallel repetition would hold. Raz's counterexample does not provide a negative
answer to this question, since it is an XOR game, for which perfect parallel repetition holds in the entangled case.
Similarly, in the case of non-signaling provers there has been no evidence that strong parallel repetition does not
hold. In fact, because the non-signaling value is {\em exactly} given by a linear program (LP) (see, e.g.,
\cite{Toner:LP}), one might conjecture that strong parallel repetition should hold since ``all" one has to do is
understand the tensorization properties of the corresponding LP.

\paragraph{Our results:}
We answer the above question in the negative, by giving a counterexample to strong parallel repetition for games
with entangled provers. More precisely, we give a game with entangled value $1-\Omega(1/n)$ such that after $n^2$
repetitions the entangled value of the repeated game is still a positive constant.
Our example (after a minor modification) is a {\em unique} game with three possible answers,
the smallest possible alphabet size for such a counterexample, because unique games with two
answers are by definition XOR games for which perfect parallel repetition holds.
Hence we obtain an interesting `phase transition' in the entangled value
of unique games: whereas for alphabet size $2$ we have \emph{perfect} parallel repetition,
already for alphabet size $3$ we do not even have strong parallel repetition.
Our result shows that the upper bound for unique games in Theorem~\ref{th:parallelforentangled}.2
is essentially tight.

We also show that our game has a non-signaling value of $1-\Omega(1/n)$.
This implies that strong parallel repetition fails also for the non-signaling value and that
Holenstein's result (Theorem~\ref{th:parallelforentangled}.4) is in fact tight.

As part of the proof we observe (see Theorem~\ref{thm:sdp1amortized}) using results from~\cite{KempeRT:unique}
that the asymptotic behavior of the entangled value
of repeated unique games is almost precisely captured by a certain SDP
(SDP1 in Sec.~\ref{sec:prelim}). This is a pleasing
state of affairs, since we now have a nearly tight SDP characterization
both of the value of a unique game (SDP2 in Sec.~\ref{sec:prelim}) and of its
asymptotic value (SDP1). Incidentally, SDP1 was also shown
to characterize the asymptotic behavior of the \emph{classical} value
of repeated unique games, although the bounds there
were considerably less tight, as they include some logarithmic factors (which are
conjectured to be unnecessary) and also depend on the
alphabet size (see Lemma~\ref{lem:crounding}).

Combining the above observation with our counterexample, we obtain a separation between SDP1 and SDP2. Namely, for the
game described in our counterexample, SDP2 is $1-\Theta(1/n)$ (since it is very close to the value of the game) whereas
SDP1 is $1-\Theta(1/n^2)$ (since it describes the asymptotic behavior). Both SDPs have been used before in the
literature (e.g., \cite{KhotV05,AroraKKSTV08,KempeRT:unique,BarakHHRRS08}) and to the best of our knowledge no gap
between them was known before. Perhaps more interestingly, our example also implies that SDP2 does \emph{not}
tensorize, since for the basic game SDP2 is $1-\Theta(1/n)$ yet after $n^2$ repetition its value is still some positive
constant (since it is a relaxation of the entangled value).

\paragraph{Our construction:}
Our counterexample is inspired by the odd-cycle game (yet it is neither a cycle nor is it odd). We call it the
\emph{line game}. Recall that the odd-cycle game was used by Raz~\cite{Raz08} as a counterexample to strong parallel repetition in the
classical case. However, since it is an XOR game, it obeys perfect parallel repetition in the entangled case,
and moreover, its non-signaling value is $1$, so it cannot provide a counterexample in our setting.

Roughly speaking, in the line game the players are asked to color a path of length $n$ with two colors in such a way that
any two adjacent vertices have the same color, yet the leftmost vertex must be colored in color 1 and the rightmost
vertex must be colored with color 2 (see Fig.~\ref{fig:linegame}a). More precisely, the verifier randomly chooses to send to the provers either two adjacent vertices
or the same two vertices. He expects the two answers to be the same, unless both vertices are the leftmost vertex,
in which case both answers must be $1$, or both vertices are the rightmost vertex, in which case both answers
must be $2$.

It is not hard to see that the classical value of this game is $1-\Theta(\frac{1}{n})$, as is the
case for the odd-cycle game. However, unlike the odd-cycle game, it turns out that the entangled value
and even the non-signaling value of this game are also $1-\Theta(\frac{1}{n})$. An intuitive way to see this is to argue about the marginals on Alice's and
Bob's answer to each question. Forcing the ends of the line into a fixed answer forces the corresponding marginals
to be close to distributions that always output $1$ on the left and $2$ on the right. The marginals for questions in
between the ends must therefore move from the all-$1$ to the all-$2$ distribution, which can only be done at the
expense of losing with probability $\Omega(1/n)$. For comparison, in the odd-cycle game we can manage with a strategy
where all marginals are uniform, and hence its non-signaling value is actually exactly $1$!

As we will show in Section~\ref{sec:linerepeat}, after repeating the line game $n^2$ times, its
entangled value (and even classical value) are still bounded from below by some positive constant.
In particular, this implies that strong parallel repetition does not hold for the entangled value
nor for the non-signaling value.
This lower bound can be shown directly by explicitly demonstrating the provers' strategy.
Instead, we will follow a slightly indirect route, using SDP\ref{sdp:1} to argue about the
behavior of the game (or in fact of its unique game variant described below) under parallel
repetition, as we feel this gives more insight into the behavior of parallel repetition of unique games.

As described above, the line game is not a unique game, due to the non-permutation constraints on both ends.
In order to provide a counterexample for strong parallel repetition even for
unique games, we present a simple modification of the game that leads to what we
call the \emph{unique line game}. Roughly speaking, this is done by increasing the answer size to $3$,
replacing the constraint on the leftmost vertex with a permutation that switches 2 and 3,
and similarly replacing the constraint on the rightmost vertex with a permutation that switches 1 and 3.
This has a similar effect to the non-permutation constraints in the original line game,
and as a result, the classical, entangled, and non-signaling values of this game are
more or less the same as in the line game, both for the basic game and its repetition.

\section{Preliminaries}\label{sec:prelim}


\paragraph{Games:}
We study \textit{one-round two-prover cooperative games of incomplete information}, also known in the quantum
information literature as \textit{nonlocal games}. In such a game, a referee (also called the verifier) interacts with
two provers, Alice and Bob, whose joint goal is to maximize the probability that the verifier outputs {\em ACCEPT}.
In more detail, we represent a game $G$ as a distribution over triples $(s,t,\pi)$ where $s$ and $t$ are elements
of some question set $Q$, and $\pi:[k]\times [k] \to \{0,1\}$ is a predicate over pairs of answers taken
from some alphabet $[k]$. The game described by such a $G$ is as follows.

\begin{itemize}
 \item The verifier samples $(s, t, \pi)$ according to $G$.
 \item He sends $s$ to Alice and receives an answer $a \in [k]$.
 \item He sends $t$ to Bob and receives an answer $b \in [k]$.
 \item He then accepts iff $\pi(a,b)=1$.
\end{itemize}

This definition of games is the one used by~\cite{BarakHHRRS08} and is slightly more
general than the one commonly used in the literature, which requires
that each pair $(s,t)$ is associated with exactly one predicate $\pi$.
Our definition allows the verifier to associate more than one predicate $\pi$ (in fact, a distribution over predicates)
to each question pair $(s,t)$. Such games are sometimes known as games with
probabilistic predicates. We use this definition mostly for convenience,
since as we shall see later, our counterexamples either do not use probabilistic predicates,
or can be modified to avoid them (but see Remark~\ref{rem:paralleledges}
for one instance in which probabilistic predicates are provably necessary).
Moreover, the results in~\cite{CleveSUU07,KempeRT:unique,BarakHHRRS08} hold
for games with probabilistic predicates, and this is in particular
true for Lemmas~\ref{lem:tensor}, \ref{lem:rounding}, and \ref{lem:SDP2rounding},
which we need for our construction. Finally, Raz~\cite{Raz98} briefly discusses
how to extend his parallel repetition theorem to games with probabilistic predicates,
whereas the results in~\cite{Holenstein07,Rao:parallel} most likely also extend to this case, although
this remains to be verified; in any case, these results are not needed for our construction.

We define the \emph{(classical) value} of a game, denoted by $\omega(G)$, to be the maximum probability
with which the provers can win the game, assuming they behave classically, namely,
they are simply functions from $Q$ to $[k]$. We can also allow the provers to share
randomness, but it is easy to see that this does not increase their winning probability.
We define the \emph{entangled value} of a game, $\omega^*(G)$, to be the maximum winning
probability assuming the provers are allowed to share entanglement. The precise definition of entangled strategies can be found in, e.g.,
\cite{KempeRT:unique}, but will not be needed in this paper.
We essentially just have to know that the entangled value is bounded from above by the \emph{non-signaling value}, which is defined as follows.

\begin{definition}\label{def:nonsignaling}
  A non-signaling strategy for a game $G$ is a set of probability distributions $\{p_{s,t}\}$ over
  $[k]\times [k]$ for all $s,t \in Q$ such that
  $$ \forall s,s',t,t'\in Q \quad A_{s,t}=A_{s,t'}=:A_s \quad \textrm{and}
  \quad B_{s,t}=B_{s',t}=:B_t,$$
  where $A_{s,t}(a)$, $B_{s,t}(b)$ are the marginals of $p_{s,t}$ on the first and second
  answer respectively. The non-signaling value of the game is
    $$\omegans(G)=\max  \displaystyle \mathop{ \Exp}_{(s,t,\pi)\sim G} \Big[ \sum_{a,b=1}^k p_{s,t}(a,b) \pi(a,b) \Big] $$
  where the maximum is taken over all non-signaling strategies $\{p_{s,t}\}$.
\end{definition}

\begin{definition}
A game is called {\em unique} if the third component of the triples $(s,t,\pi)$ is always
a permutation constraint, namely, it is $1$ iff $\sigma(a)=b$ for some permutation $\sigma$.
We will sometimes think of such games as distributions over triples $(s,t,\sigma)$.
Furthermore, a unique game is called \emph{linear} if we can identify $[k]$ with some
Abelian group of size $k$ and the third component of $(s,t,\sigma)$
is always of the form $\sigma(a)=a+r$ for some element $r$ of the group.
\end{definition}

\paragraph{Parallel Repetition:} Given a game $G_1$ with questions $Q_1$ and answers in $[k_1]$ and the game $G_2$ with questions $Q_2$ and answers in $[k_2]$,
we define the {\em product} $G_1 \times G_2$ to be a game with questions $Q_1 \times Q_2$ and answers in $[k_1] \times
[k_2]$ defined by the distribution obtained by sampling $(s_1,t_1,\pi_1)$ from $G_1$
and $(s_2,t_2,\pi_2)$ from $G_2$ and outputting $((s_1,s_2),(t_1,t_2),\pi_1\times \pi_2)$ where
$\pi_1 \times \pi_2: [k]^2 \times [k]^2 \to \{0,1\}$ is given by $(\pi_1 \times \pi_2)((a_1,a_2),(b_1,b_2)) =
\pi_1(a_1,b_1) \pi_2(a_2,b_2)$.

We denote the $\ell$-fold product of $G$ with itself by $G^{
\ell}$. Clearly, $\omega(G^\ell) \ge \omega(G)^\ell$ and similarly for $\omega^*$ and $\omegans$, since
the provers can play each instance of the game independently, using an optimal strategy. Parallel repetition theorems attempt to provide upper bounds on
the value of repeated games. It is often convenient to speak about the {\em amortized value} of a game,
defined as $\overline{\omega}(G)=\lim_{\ell \rightarrow \infty} \omega(G^\ell)^{\frac{1}{\ell}} \ge \omega(G)$, and similarly for $\overline
\omega^*(G)$ and $\overline \omegans(G)$.

\paragraph{SDP Relaxations:}
The main SDP relaxation we consider in this paper is SDP\ref{sdp:1}, which is defined for any game $G$.
The maximization is over the real vectors $\{u_a^s\}$, $\{v_b^t\}$.

\begin{program}
\caption{}\label{sdp:1}
\begin{tabular}{p{0.24 \textwidth}p{0.7 \textwidth}}
\textbf{Maximize:} & $\Exp_{(s,t,\pi)\sim G} \sum_{ab} \pi(a,b) \ip{u_a^s, v_b^t}$\\
\textbf{Subject to:} & $\forall s,~\forall a \neq b,~ \ip{u_a^s,u_b^s}=0$ and $\forall t,~\forall a \neq b,~ \ip{v_a^t,v_b^t}=0$\\
& $\forall s,~\sum_{a}\ip{u_a^s,u_a^s} =1$ and $\forall t,~\sum_{b}\ip{v_b^t,v_b^t} =1$\\
\end{tabular}
\end{program}

It follows from Theorem 5.5 and Remark 5.8 of~\cite{KempeRT:unique}
that SDP\ref{sdp:1} has the tensorization property.
 \begin{lemma}\label{lem:tensor}
For any game $G$ and any $\ell \ge 1$, $\omega_{SDP1}(G^\ell)=(\omega_{SDP1}(G))^\ell$,
where $\omega_{SDP1}$ denotes the optimum value of $SDP1$ for a particular game.
 \end{lemma}

The proof of this lemma is based on ideas from \cite{FL,MittalSzegedy:SDPtensor}.
Ignoring some subtle issues,
the essential reason that Lemma~\ref{lem:tensor} holds is because
SDP\ref{sdp:1} is \emph{bipartite}, i.e., its goal function only involves inner products
between $u$ variables and $v$ variables, and its constraints
are all equality constraints and involve either only $u$ variables or only $v$ variables
(see~\cite{KempeRT:unique} for details).

The value of SDP\ref{sdp:1} is an upper bound for the entangled value of the game, and in \cite{KempeRT:unique} it is shown that its value is not too far from
the entangled value of unique games.

\begin{lemma}[\cite{KempeRT:unique}]\label{lem:rounding}
Let $G$ be a unique game with $\omega_{SDP1}(G)=1-\delta$.
Then $1-8 \sqrt{\delta} \le \omega^*(G) \le 1- \delta$.
\end{lemma}

Moreover, in a recent result by Barak et al.~\cite{BarakHHRRS08} it was shown that SDP\ref{sdp:1} essentially
characterizes the amortized (classical) value of unique games, up to a factor that depends on the alphabet size and logarithmic corrections.
\begin{lemma}[\cite{BarakHHRRS08}]\label{lem:crounding}
For any unique game $G$ with $\omega_{SDP1}(G)=1-\delta$ and $\ell \ge 1$,
$\omega(G^\ell)\geq 1-O(\sqrt{\ell \delta \log(k/\delta)})$, and moreover,
$1-O(\delta \log(k/\delta)) \le \overline{\omega}(G^\ell) \le 1-\delta$.
\end{lemma}

\begin{program}
\caption{}\label{sdp:2}
\begin{tabular}{p{0.12 \textwidth}p{0.8 \textwidth}}
\textbf{Maximize:} & $\Exp_{(s,t,\pi)\sim G} \sum_{ab} \pi(a,b) \ip{u_a^s, v_b^t}$\\
\textbf{Subject to:} & $\|z\|=1$\\
& $\forall s,t,~\sum_a u_a^s = \sum_b v_b^t = z$\\
& $\forall s,t,~\forall a \neq b,~ \ip{u_a^s,u_b^s}=0$ and $\ip{v_a^t,v_b^t}=0$\\
& $\forall s,t,a,b,~\ip{u_a^s,v_b^t} \ge 0$\\
\end{tabular}
\end{program}

We now consider SDP\ref{sdp:2}. Notice the extra variable $z$,
the extra non-negativity constraints, and the extra $z$ constraints.
We clearly have that for any game $G$, $\omega_{SDP2}(G) \le \omega_{SDP1}(G)$.
Yet, as mentioned in~\cite{KempeRT:unique}, SDP2 still provides an upper bound
on the entangled value. Moreover, for unique games, this upper bound is almost tight.

\begin{lemma}[\cite{KempeRT:unique}]\label{lem:SDP2rounding}
Let $G$ be a unique game with $\omega_{SDP2}(G)=1-\delta$. Then $ 1-6 \delta \le \omega^*(G) \le  1-\delta$.
\end{lemma}

It was not known whether SDP2 satisfies the tensorization property.

\section{The line game and its non-signaling value}\label{sec:line}

We now describe and analyze our first counterexample, the line game (see Figure~\ref{fig:linegame}a for
an illustration).

\begin{figure}[h]
\begin{centering}
\includegraphics[width=0.7\textwidth]{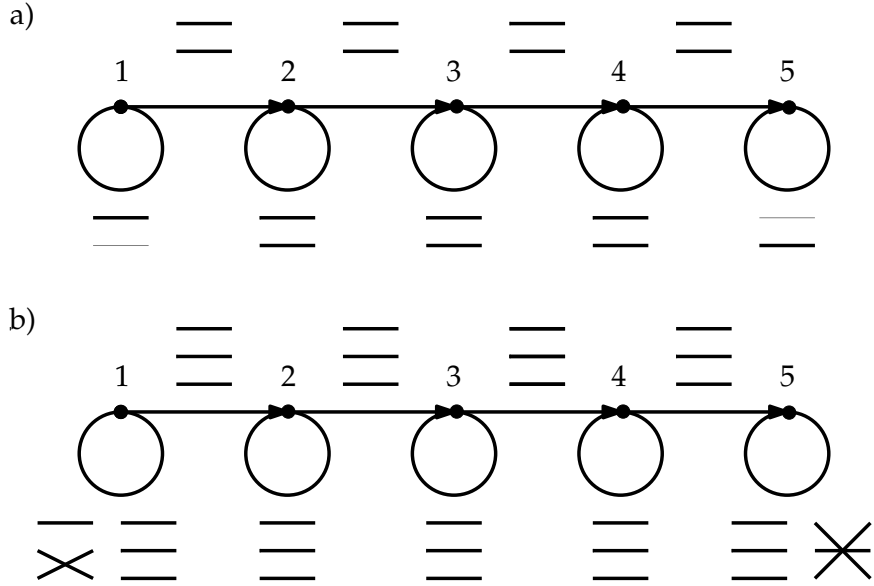}
\caption{The line game (top) and the unique line game (bottom) for $n=5$. }
\label{fig:linegame}
\par\end{centering}
\end{figure}

\begin{definition}[Line game]
Consider a path with vertices $\{1,\ldots,n\}$ with edges connecting any two successive nodes,
as well as a loop on each vertex (so the total number of edges is $2n-1$).
The \emph{line game} $G_L$ of length $n$ is a game with question set $Q=[n]$,
and answer size $k=2$, in which the verifier chooses a triple $(s,t,\pi)$
as follows. He first chooses an edge with endpoints $s \le t$ uniformly
among the $2n-1$ edges. The constraint $\pi$ is set to be equality for all edges,
except for the two loops at the ends, i.e., except $s=t=1$ or $s=t=n$.
In the former case, the constraint $\pi$ forces $a=b=1$ and in
the latter case it forces $a=b=2$.
\end{definition}
Note that the line game $G_L$ is not a unique game due to the non-unique constraints
on both ends of the line.

\begin{theorem}\label{thm:nsvalueL}
$\omega(G_L)=\omega^*(G_L)=\omegans(G_L) = 1-\frac{1}{2n-1}$.
\end{theorem}

\begin{proof}
First, notice that the success probability of the classical strategy
in which Alice and Bob always answer $1$ is $1-\frac{1}{2n-1}$.
Hence, $1-\frac{1}{2n-1} \le \omega(G_L) \le \omega^*(G_L)\le \omegans(G_L)$
and it remains to bound $\omegans(G_L)$ from above. For this, we use the following simple claim.
\begin{claim}\label{clm:delta}
For any $k \ge 1$, $a,b \in [k]$, any permutation $\sigma$ on $[k]$ and any probability distribution
$p$ on $[k] \times [k]$ with marginal distributions $A(a)$ and $B(b)$,
\begin{equation*}
\Pr_{(a,b) \sim p}[a=\sigma(b)]\leq 1-\Delta(A,\sigma(B)),
\end{equation*}
where $\Delta(A,\sigma(B))=\frac{1}{2}\sum_{a \in [k]}|A(a)-B(\sigma(a))|$ is the total variation distance between
$A$ and $\sigma(B)$. Moreover, for any marginal distributions $A$ and $B$ there exists a distribution $p$
for which equality is achieved.
\end{claim}
\begin{proof}
For simplicity assume that $\sigma$ is the identity permutation;
the general case follows by permuting the answers. Note that $p(a,a) \leq
\min(A(a),B(a))=\frac{1}{2}(A(a)+B(a)-|A(a)-B(a)|)$. Hence $$\Pr_{(a,b)\sim p}[a=b]=\sum_{a \in [k]} p(a,a)\leq
\frac{1}{2}\sum_{a \in [k]} \Big( A(a)+B(a)-|A(a)-B(a)| \Big) =1-\Delta(A,B).$$  To construct a $p$ such that equality holds, we can
simply set $p(a,a)=\min(A(a),B(a))$. It is easy to see that it is possible to complete this to a probability
distribution.
\end{proof}
We now bound the non-signaling value of $G_L$ by arguing about the marginal distributions of the provers' strategy.
Let $\{p_{s,t}|s,t \in [n]\}$ be an arbitrary non-signaling strategy,
let $A_1,\ldots , A_n$ be the marginal distributions on Alice's answers and $B_1,\ldots
,B_n$ the marginal distributions for Bob, as in Def.~\ref{def:nonsignaling}. Note that  except for question pairs $(1,1)$, $(n,n)$
all constraints are equality constraints. Hence, using Claim~\ref{clm:delta} and denoting the number of edges by $m=2n-1$,
the winning probability for this strategy is at most
\begin{align}
&1-\frac{2}{m}-\frac{1}{m}\left(\sum_{s=2}^{n-1}\Delta(A_s,B_s)+\sum_{s=1}^{n-1}
\Delta(A_s,B_{s+1})\right)+\frac{1}{m}(p_{1,1}(1,1)+p_{n,n}(2,2))\nonumber \\
&\leq 1-\frac{2}{m}- \frac{1}{m} \Delta(A_1,B_n) +\frac{1}{m} (p_{1,1}(1,1)+p_{n,n}(2,2)) \nonumber \\
&= 1-\frac{2}{m} +  \frac{1}{m} \big( p_{1,1}(1,1)+p_{n,n}(2,2) - \Delta(A_1,B_n) \big), \label{eq:inthere}
\end{align}
where in the first inequality we used the triangle inequality for total variation distance $\Delta$.
We complete the proof by noting that $\Delta(A_1,B_n) \ge A_1(1) + B_n(2) - 1$, and
recalling that by definition $A_1(1) \ge p_{1,1}(1,1)$ and $B_n(2) \ge p_{n,n}(2,2)$.
\end{proof}

In order to show that the line game violates strong parallel repetition we will modify it to a unique game $G_{uL}$
by increasing the alphabet size to $3$ and slightly changing the constraints.
We will shortly see that $G_L$ and $G_{uL}$ have essentially the same
non-signaling value and behave similarly under parallel repetition.

\begin{definition}[Unique line game]
Consider a path with vertices $\{1,\ldots,n\}$ with edges connecting any two successive nodes,
as well as a loop on each vertex (so the total number of edges is $2n-1$).
The \emph{unique line game} $G_{uL}$ of length $n$ is a game with question set $Q=[n]$,
and answer size $k=3$, in which the verifier chooses a triple $(s,t,\sigma)$
as follows. He first chooses an edge with endpoints $s \le t$ uniformly
among the $2n-1$ edges. The permutation $\sigma$ is set to be the identity for all edges,
unless $s=t=1$ or $s=t=n$.
In the former case, $\sigma$ is chosen to be the identity with probability half
and the permutation that switches 2 and 3 with probability half;
in the latter case, $\sigma$ is chosen to be the identity with probability half
and the permutation that switches 1 and 3 with probability half.
\end{definition}

\begin{theorem}\label{thm:nsvalueuL}
$\omega(G_{uL})=\omega^*(G_{uL})=\omegans(G_{uL}) = 1-\frac{1}{2(2n-1)}$.
\end{theorem}
\begin{proof}
First, the strategy that assigns answer $1$ to all questions achieves winning probability $1-\frac{1}{2(2n-1)}$.
The upper bound can be shown by repeating the proof of Thm.~\ref{thm:nsvalueL} with minor
modifications. Instead, let us show how to obtain the upper bound as a corollary to
Thm.~\ref{thm:nsvalueL}. Let $\{p_{s,t}|s,t \in [n]\}$ be an arbitrary non-signaling strategy
for $G_{uL}$, and consider the strategy obtained by mapping answer 3 to 2. More precisely,
define $\{\tilde{p}_{s,t}|s,t \in [n]\}$ as the strategy with alphabet size $2$
defined by $\tilde{p}_{s,t}(1,1)=p_{s,t}(1,1)$, $\tilde{p}_{s,t}(1,2)=p_{s,t}(1,2)+p_{s,t}(1,3)$,
$\tilde{p}_{s,t}(2,1)=p_{s,t}(2,1)+p_{s,t}(3,1)$, and
$\tilde{p}_{s,t}(2,2)=p_{s,t}(2,2)+p_{s,t}(3,2)+p_{s,t}(2,3)+p_{s,t}(3,3)$.
Then it is easy to check that $\tilde{p}$ is also a non-signaling strategy
and that moreover, its value under the game $G_{L}$ is at least the average
between 1 and the value of $p$ under $G_{uL}$. Hence
$\omegans(G_{uL}) \le 2\omegans(G_{L}) - 1$, as desired.
\end{proof}

\begin{remark}\label{rem:paralleledges}
Notice that the unique line game uses probabilistic predicates, i.e., there are questions
(namely, the two end loops) to which more than one predicate is associated. It is not
difficult to avoid these probabilistic predicates by replacing the end loops with small gadgets,
while keeping the \emph{classical} and \emph{entangled} values of the game as well as those of the repeated
game more or less the same, hence leading to a counterexample to strong parallel repetition
using unique games with deterministic predicates
(namely, just add one extra vertex at each end, call it $1'$ and $n'$ and add equality constraints
from $1$ to $1'$ and from $1'$ to $1'$, as well as a constraint that switches $2$ and $3$ from $1'$ to $1$, and
analogous modification for $n'$). However, note that it is impossible
to obtain a counterexample to strong parallel repetition for the \emph{non-signaling} value
that is both unique and uses deterministic predicates. The reason is
that \emph{any} unique game with deterministic predicates has non-signaling
value $1$: simply choose for each question pair $(s,t)$ the
distribution $p_{s,t}(a,b)=\frac{1}{k}$ if $a=\sigma_{st}(b)$ and $p_{s,t}(a,b)=0$ otherwise, where
$\sigma_{st}$ is the unique permutation associated with $(s,t)$. This strategy is
non-signaling, as all its marginal distributions are uniform. Hence any {\em
unique} game that gives a counterexample to strong parallel repetition in the non-signaling case
must use probabilistic predicates.
\end{remark}

\section{Parallel repetition of the line game}\label{sec:linerepeat}

We now proceed to show that strong parallel repetition holds neither for $G_{uL}$ nor for $G_L$.
We will show this by
first proving a general connection for unique games between the value of SDP\ref{sdp:1} and the repeated entangled
value of the game. We emphasize that the following construction can also be presented more explicitly
without resorting to SDPs; we feel, however, that the connection to SDPs gives much more insight into
the nature of parallel repetition, and might also make it easier to extend our result to other settings.

\begin{theorem}\label{thm:sdp1amortized}
For any unique game $G$, if $\omega_{SDP1}(G)=1-\delta$ then

 (i) for all $\ell \geq 1$ we have
$1-8\sqrt{\ell\delta} \le \omega^*(G^\ell) \leq (1-\delta)^\ell$ and

(ii)
for all $\ell>\frac{1}{\delta}$ we have $ (1-c\delta)^{\ell} \leq \omega^*(G^\ell) \leq (1-\delta)^\ell$ for some
universal constant $c>0$.\\ In particular, the amortized entangled value is $\overline{\omega^*}(G) = 1-\Theta(\delta)$.
\end{theorem}

Compare this to the classical case (Lemma~\ref{lem:crounding}), where we have a dependence on the alphabet size
(as well as an extra $\log$ factor). In the
entangled case, SDP\ref{sdp:1} gives a tight estimate on the amortized entangled value up to a universal constant.
\begin{proof}
We combine several statements from~\cite{KempeRT:unique}. By Lemma \ref{lem:tensor},
$\omega_{SDP1}(G^\ell)=\omega_{SDP1}(G)^\ell=(1-\delta)^\ell\geq 1-\ell
\delta$. We now use the quantum rounding of \cite{KempeRT:unique}, Lemma \ref{lem:rounding} to obtain an entangled
strategy for $G^\ell$ with value at least $1-8 \sqrt{\ell \delta}$, showing part (i). Part (ii) follows by
partitioning the $\ell$ repetitions into blocks of size $\frac{1}{100\delta}$ and playing the
strategy of (i) on each block independently.
\end{proof}

Hence, in order to analyze the repeated entangled value of $G_{uL}$ it suffices
to analyze its SDP1 value.

\begin{lemma}\label{lem:sdpline}
$\omega_{SDP1}(G_{uL})\geq 1-\frac{2}{n^2}$.
\end{lemma}

\begin{proof}
We construct a solution $\{u^s_a\},\{v^t_b\} \in \mathbb{R}^2$ for SDP1$(G_{uL})$ in the following way, as illustrated
in Fig.~\ref{fig:lineSDP}a.

\begin{figure}[h]
\begin{centering}
\includegraphics[width=0.7\textwidth]{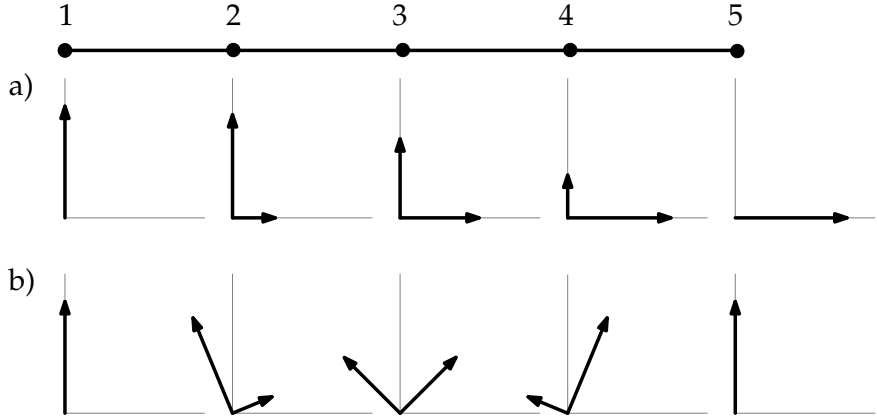}
\caption{Two SDP solutions} \label{fig:lineSDP}
\par\end{centering}
\end{figure}
\begin{equation}\label{eq:SDP1solution}
\forall s \in \{1,\ldots ,n\} \quad u^s_1=v^s_1=\left(\begin{array}{c}0 \\ \cos{\frac{s-1}{n-1}\frac{\pi}{2}}
\end{array} \right), \quad u^s_2=v^s_2=\left(\begin{array}{c}\sin{\frac{s-1}{n-1} \frac{\pi}{2}} \\ 0  \end{array}
\right),  \quad u^s_3=v^s_3=0
\end{equation}
Clearly $\ip{{u^s_a},{u^s_b}}=0$ for $a \neq b$ and $\sum_a \ip{u^s_a,u^s_a}=1$ and similarly for the $v$ vectors, so
our solution for SDP1($G_{uL})$ is {\em feasible}. Since the $u$ vectors are equal to the $v$ vectors, it is easy to
compute its value
\begin{align*}
&\Exp_{(s,t,\pi)\sim G} \sum_{ab} \pi(a,b) \ip{u_a^s, v_b^t}  \\
&=\frac{1}{2n-1}\Big(\ip{u_1^1, v_1^1}+\ip{u_2^n,v_2^n}+\sum_{s=2}^{n-1}\Big(\ip{u_1^s, v_1^s}+\ip{u_2^s, v_2^s}  \Big)\Big)
+ \frac{1}{2n-1} \sum_{s=1}^{n-1} \Big( \ip{u_1^s, v_1^{s+1}}+\ip{u_2^s, v_2^{s+1}}   \Big) \\
&=\frac{n}{2n-1}+\frac{1}{2n-1}\sum_{s=1}^{n-1} \Big( \cos \Big(\frac{\pi}{2}\frac{s-1}{n-1}\Big) \cos \Big(\frac{\pi}{2}\frac{s}{n-1}\Big)+\sin
\Big(\frac{\pi}{2}\frac{s-1}{n-1}\Big) \sin \Big(\frac{\pi}{2}\frac{s}{n-1}\Big) \Big) \\
&=\frac{n}{2n-1}+\frac{n-1}{2n-1}\cos \frac{\pi}{2(n-1)}\geq
\frac{n}{2n-1}+\frac{n-1}{2n-1}(1-\frac{\pi^2}{8(n-1)^2})=1-\frac{\pi^2}{8(n-1)(2n-1)},
\end{align*}
which proves the lemma for all $n \geq 2$.
\end{proof}

Combining the above lemma with Lemma~\ref{lem:rounding}, we see that in fact $\omega_{SDP1}(G_{uL}) =  1-\Theta(\frac{1}{n^2})$.
Moreover, Lemma~\ref{lem:SDP2rounding} shows that $\omega_{SDP2}(G_{uL}) =  1-\Theta(\frac{1}{n})$.
Hence we obtain a quadratic gap between SDP1 and SDP2. Also note that the SDP1 solution above
obeys the non-negativity constraint in SDP2: the inner products of any two vectors is non-negative.
In fact we can modify the solution to a solution with similar value, so that it obeys the $z$-constraint of SDP2, at the expense of
violating the non-negativity constraint, as shown in Fig.~\ref{fig:lineSDP}b. 
Hence our quadratic gap also holds between SDP2 and the two possible strengthenings of SDP1.

Combining Theorem~\ref{thm:sdp1amortized} with the above lemma, we obtain that
for all $\ell \geq n^2$, $\omega^*(G_{uL}^\ell) \geq (1-O(1/n^2))^\ell$.
In fact, the same lower bound also holds for the \emph{classical} value.
The reason for this is that the strategy constructed in Lemma~\ref{lem:rounding}
uses a shared maximally entangled state, and performs a measurement on it
in an orthonormal basis derived from the SDP vectors. Since all the vectors in
the SDP solution of Lemma~\ref{lem:sdpline} are in the same orthonormal basis
(and the same is true for the resulting SDP solution of $G_{uL}^\ell$),
we obtain that the strategy constructed in Lemma~\ref{lem:rounding} is
in fact a classical strategy. A final technical remark is that even though
we obtained the above strategy by using a tensored SDP solution, the
strategy itself is not a product strategy due to a ``correlated
sampling" step performed as part of the proof of Lemma~\ref{lem:rounding}.
We summarize this discussion in the following
theorem.

\begin{theorem}\label{thm:mainunique}
For $\ell \geq n^2$, $\omegans(G_{uL}^\ell) \ge \omega^*(G_{uL}^\ell) \ge  \omega(G_{uL}^\ell) \ge (1-O(1/n^2))^\ell$.
\end{theorem}

This shows that Holenstein's parallel repetition for the non-signaling value (Theorem~\ref{th:parallelforentangled}.4)
as well as the parallel repetition theorem for the entangled value of unique games
(Theorem~\ref{th:parallelforentangled}.2) are both tight up to a constant.%
\footnote{Strictly speaking, Holenstein's proof does not deal with probabilistic predicates,
although it can most likely be extended to deal with this case~\cite{Holenstein:private},
as was done in~\cite{Raz98}. In any case, the line game (which we consider next)
gives an alternative tight example for Holenstein's theorem with deterministic predicates.}

We complete this section by extending the above analysis to the line game, as shown
in the following theorem. This shows that alphabet size $2$ is sufficient to obtain
a counterexample to strong parallel repetition for both the entangled value and the non-signaling value,
and in particular shows that Theorem~\ref{th:parallelforentangled}.4 is tight also for this case.
The counterexample is not a unique game, but this is actually necessary: XOR games obey perfect parallel
repetition both in terms of the entangled value (Thm.~\ref{th:parallelforentangled}.1) and in terms of the non-signaling value
(even with probabilistic predicates, as is not difficult to see).

\begin{theorem}\label{thm:main}
For $\ell \geq n^2$, $\omegans(G_{L}^\ell) \ge \omega^*(G_{L}^\ell) \ge  \omega(G_{L}^\ell) \ge (1-O(1/n^2))^\ell$.
\end{theorem}
\begin{proof}
We first observe that the classical strategy for $G_{uL}^\ell$ constructed above
has the property that both provers always answer $1$ on a coordinate containing the question $1$, and
similarly, they always answer $2$ on a coordinate containing the question $n$.
Moreover, the provers never use the answer $3$.
This follows from the fact that the vectors constructed in Lemma~\ref{lem:sdpline}
satisfy $u_2^1=v_2^1=0$, $u_1^n=v_1^n=0$, and for all $s$, $u_3^s=v_3^s=0$.
As a result, when taking the tensor product of these vector and applying Lemma~\ref{lem:rounding} (as was
done in Theorem~\ref{thm:sdp1amortized}), we obtain the aforementioned property of the classical
strategy. Since the strategy does not use the answer $3$, it is also a valid strategy for
$G_{L}^\ell$. Moreover, it is easy to check that the winning probability of the strategy
in $G_{L}^\ell$ is equal to that in $G_{uL}^\ell$; this is because the strategy always
answers $1$ on $1$ and $2$ on $n$, and due to the way the games are constructed.
\end{proof}

\paragraph{Acknowledgments:}
We thank Thomas Holenstein for answering our queries regarding his parallel repetition theorem, and Nisheeth Vishnoi
for useful discussions.

\newcommand{\etalchar}[1]{$^{#1}$}

\end{document}